\documentclass[table]{article}
\usepackage{spconf,amsmath,graphicx}

\usepackage[utf8]{inputenc}
\usepackage{amsfonts}
\usepackage[spanish, english]{babel}
\usepackage{tcolorbox}
\usepackage{amsthm}  
\usepackage{mathrsfs}
\usepackage{xparse}
\usepackage{algorithm}
\usepackage{algorithmic}
\tcbuselibrary{skins}
\usepackage{color}
\usepackage{makecell, cellspace}
\usepackage{setspace}
\usepackage{graphicx}
\usepackage{multirow}
\definecolor{Pastel}{RGB}{245,245,250}
\definecolor{Pastel2}{gray}{0.99}
\usepackage{verbatim}
\makeatletter
\newcommand*\bigcdot{\mathpalette\bigcdot@{.5}}
\newcommand*\bigcdot@[2]{\mathbin{\vcenter{\hbox{\scalebox{#2}{$\m@th#1\bullet$}}}}}
\makeatother
\usepackage{bibentry}
\usepackage{makecell, cellspace}
\usepackage{booktabs}
\usepackage{appendix}
\usepackage[normalem]{ulem}
\usepackage{mathtools, physics}

\usepackage{caption}
\usepackage[subrefformat=simple, labelformat=simple]{subcaption}

\usepackage{flushend}

\newcommand{\bt}[1]{\mbox{$\bf #1$}}
\def\l{\left(}
\def\r{\right)}

\makeatletter
\def\thm@space@setup{\thm@preskip=6pt
\thm@postskip=6pt}
\makeatother

\newtheorem{theorem}{Theorem}[section] 
\newtheorem{lemma}[theorem]{Lemma}     
\newtheorem{corollary}[theorem]{Corollary} 
\newtheorem{proposition}[theorem]{Proposition} 
\newtheorem{definition}[theorem]{Definition} 

\newcounter{remark}

\newenvironment{remark}[1][]{\refstepcounter{remark}\medskip
   \noindent\textbf{Remark~\theremark.#1} \rmfamily\ignorespaces}{\medskip}
\usepackage{enumitem}

\usepackage{amsmath}

\usepackage{amsmath}
\usepackage{siunitx}
\newcommand{\baseL}{\tilde{\bt L}}

\newcommand{\bfm}{\pmb{\mu}}
\newcommand{\bfl}{\pmb{\lambda}}

\newcommand{\tbfm}{\tilde{\bfm}}
\newcommand{\tbfl}{\tilde{\bfl}}

\newcommand{\Cml}{\bt C(\bfm, \bfl)}
\newcommand{\Clm}{\bt C(\bfl, \bfm)}
\newcommand{\diag}[1]{\mathrm{diag}\l \bt #1\r}

\newcommand{\Lv}{\bt L(\bt v, \rho)}
\newcommand{\defeq}{\doteq}

\title{Fast DCT+: A Family of Fast Transforms Based on \\ Rank-One Updates of the Path Graph}
\name{Samuel Fern\'andez-Mendui\~na, Eduardo Pavez, Antonio Ortega}
\address{University of Southern California, Los Angeles, CA 90089, USA}

\begin{document}
\ninept
\maketitle
\begin{abstract}
This paper develops fast graph Fourier transform (GFT) algorithms with $O(n\log n)$ runtime complexity for rank-one updates of the path graph. We first show that several commonly-used audio and video coding transforms belong to this class of GFTs, which we denote by DCT+. Next, starting from an arbitrary generalized graph Laplacian and using rank-one perturbation theory, we provide a factorization for the GFT after perturbation. This factorization is our central result and reveals a progressive structure: we first apply the unperturbed Laplacian's GFT and then multiply the result by a Cauchy matrix. By specializing this decomposition to path graphs and exploiting the properties of Cauchy matrices, we show that Fast DCT+ algorithms exist. We also demonstrate that progressivity can speed up computations in applications involving multiple  transforms related by rank-one perturbations (e.g., video coding) when combined with pruning strategies. Our results can be extended to other graphs and rank-$k$ perturbations. Runtime analyses show that Fast DCT+ provides computational gains over the naive method for graph sizes larger than $64$, with runtime approximately equal to that of $8$ DCTs.
\end{abstract}
\begin{keywords} Graph Fourier transform, fast algorithms, Cauchy matrices, rank-one, graph Laplacian, path graph.
\end{keywords}%
\section{Introduction}
\label{sec:intro}
The graph Fourier transform (GFT) \cite{ortega2018graph} allows to manipulate graph signals, much like the discrete Fourier transform (DFT) in digital signal processing (DSP). 
However, the GFT often lacks fast implementations. In DSP, the fast Fourier transform (FFT) \cite{brigham1988fast} allows computing the DFT with  $O(n\log n)$ runtime complexity, improving upon the $O(n^2)$ naive matrix-vector product (NMVP) algorithm. The FFT exploits symmetries in regular domains, but these symmetries are typically absent in graphs: exact $O(n\log n)$ algorithms for GFTs are known only for specific graph families \cite{strang1999discrete} related to either the discrete trigonometric transforms (DTTs) \cite{puschel_algebraic_2003} or circulant structures \cite{ekambaram2013circulant} for which the GFT is the DFT. For other graphs, the GFT requires NMVP with $O(n^2)$ complexity. 

Existing approaches to improve efficiency of GFT algorithms seek structured factorizations of the transform basis (e.g., via Givens rotations). Early efforts approximated the graph Laplacian \cite{le2017approximate} or its eigenvector basis \cite{frerix2019approximating} by a sequence of structured matrices. However, since spectral mismatch errors are hard to control \emph{a priori}, balancing complexity and approximation accuracy becomes challenging. Moreover, finding optimal factorizations is NP-hard \cite{frerix2019approximating} and relaxation approaches still involve significant setup costs. A partial solution \cite{lu2019fast} was proposed  for graphs with symmetries by exploiting exact butterfly-based \cite{han2013butterfly} decompositions of the transform basis. This approach can halve the number of multiplications, but the complexity remains $O(n^2)$, limiting scalability as the graph sizes grow.
\begin{figure}
    \centering
    \includegraphics[width=\linewidth]{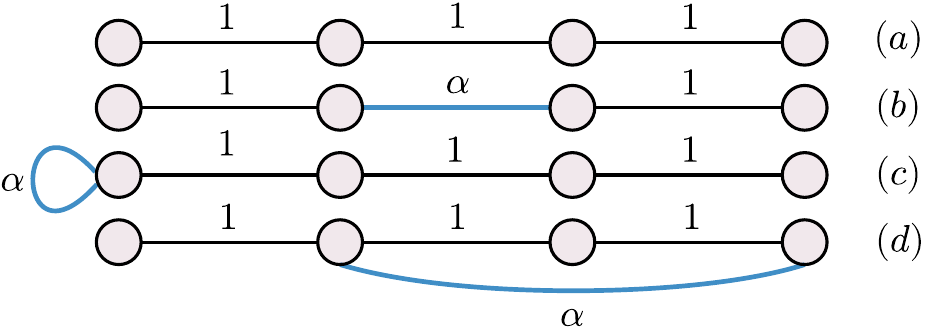}
    \caption{Path graph (a) and examples of DCT+ graphs, related to the path graph by a rank-one update of the Laplacian (b-d).}
    \label{fig:path_graph}
    \vspace{-1.5em}
\end{figure}

Hence, with existing methods one can have 1) $O(n\log n)$ algorithms that approximate the actual GFT for graphs with arbitrary topology or 2) exploit specific properties of the graph (e.g., symmetries) to have exact GFT factorizations,  which reduces computations but remains in $O(n^2)$. In this paper, we achieve both $O(n\log n)$ algorithms and exact factorizations for graphs derived from path graph  modifications \cite{strang1999discrete} (cf.~Fig.~\ref{fig:path_graph}). These graphs are often used to design transforms for audio/video coding \cite{malvar1990lapped, bross2021overview}: transforms derived from path graphs with self-loops of different strengths \cite{pavez2017learning} are implemented in the AV2 codec of the Alliance for Open Media \cite{sole_segui_communication, han2021technical}, while graphs resulting from adding, removing, or updating an edge to the path graph appear associated with Lapped transforms \cite{lu2019lapped} and graph learning \cite{lu2024adaptive, egilmez2020graph}. We unify these cases by showing that the Laplacian of all such graphs is related to the Laplacian of the path graph by a rank-one update. Since the transform for the path graph is the discrete cosine transform (DCT), we denote by DCT+ the family of transforms related to the DCT via a rank-one update of their corresponding graph Laplacian.

Our DCT+ computation approach is based on a general result: starting from the generalized Laplacian of an arbitrary graph and using rank-one perturbation theory \cite{bunch_rank-one_1978}, we derive an exact structured factorization for the eigenvector basis of the updated graph Laplacian (Theorem~\ref{thm:main_eq}). This decomposition reveals a \emph{progressive structure}: the transform for any rank-one update first applies the original unperturbed Laplacian's transform and then multiplies the result by a Cauchy matrix. Since matrix-vector products with Cauchy matrices have $O(n\log^2n)$ algorithms \cite{golub1985trummer,gerasoulis1988fast},  we show that once the signal coefficients in the unperturbed transform domain are available,   the transform of any updated Laplacian can be implemented with an additional $O(n\log^2 n)$ overhead (Theorem~\ref{th:fast_general}). Focusing on diagonal matrices, Fasino connected in a recent paper \cite{fasino_orthogonal_2023} rank-one perturbations \cite{bunch_rank-one_1978} with Cauchy matrices \cite{golub1985trummer}. Our paper extends his work to arbitrary symmetric matrices. 
 
Specializing this result to the DCT+ family of transforms, we show that the product with the Cauchy matrix can be implemented with precision $\epsilon$ and complexity $O(n\log n + n\log(1/\epsilon))$; as a result, each DCT+ has a fast implementation, or Fast DCT+ (Theorem~\ref{th:fast_line}). The algorithm can be realized using FFT-based methods with appropriate scalings, followed by a non-uniform fast Fourier transform (NUFFT) \cite{dutt1993fast}, which performs computations with a precision $\epsilon$ that can be chosen \emph{a priori}. The use of the NUFFT to compute matrix-vector products with Cauchy matrices is also novel to the best of our knowledge. Moreover, our methodology extends to DTT-related graphs\footnote{Similarly, this extension leads to the family of DTT+ transforms.}, and rank-$k$ updates.

Progressivity can be further used for rate-distortion optimization (RDO) \cite{ortega1998rate}. Modern video codecs \cite{han2021technical, bross2021overview} apply several DCT+ to intra/inter prediction residuals (e.g., DCT and ADST \cite{han2011jointly}) and then select the one optimizing an RD cost \cite{ortega1998rate}. Progressivity allows to apply first the DCT, prune the high-frequency components exploiting energy compaction in the DCT domain \cite{lengwehasatit2004scalable}, and then approximate the remaining DCT+ via Cauchy matrices, reducing the aggregated cost of computing them all in exchange for some approximation errors. Such errors depend on the signal's energy compaction in the DCT and the DCT+ domains (Theorem~\ref{th:pruning_th}).

We characterize the accuracy and runtime of Fast DCT+ considering various graph sizes and three types of rank-one updates. Results show numerically accurate transform computations (SNR above $100$ dB) and confirm the dependency on $O(n\log n)$, with runtimes roughly $8$ times larger than the DCT and speedups over NMVP for graphs with at least $64$ nodes.
We also test the pruning idea in an RDO scenario, demonstrating speedup by a factor of $1.5$. 

\medskip

\noindent \textbf{Notation.} Uppercase bold letters, such as $\bt A$, denote
matrices. Lowercase bold letters, such as $\bt a$, denote vectors.
The $n$th entry of the vector $\bt a$ is $a_n$, and the $(i, j)$th entry of the matrix $\bt A$ is $A_{ij}$. Regular letters denote
scalar values. The subvector of $\bt a$ with entries between positions $n$ and $m$ is $\bt a_{n:m}$. 

\section{Preliminaries}
\begin{definition}[Cauchy matrix \cite{gastinel1960inversion}]
Given two vectors $\bfm\in\mathbb{R}^n$ and $\bfl\in\mathbb{R}^{n}$ with no common entries, we define the Cauchy matrix $\Cml\in\mathbb{R}^{n\times n}$ as $C_{ij} = 1 / \l \mu_i - \lambda_j\r$ for $i, j = 1, \hdots n$.
\end{definition}
In our case, $\bfm$ and $\bfl$ are guaranteed to have no common entries due to the properties of the perturbed Laplacian (cf.~Lemma~\ref{lemma:inter}).
\begin{definition}[GFT \cite{ortega2018graph}]
\label{def:gft}
Let $\mathcal{G} = (\mathcal{V}, \mathcal{E}, \bt W, \bt V)$ be a weighted undirected graph, where $\mathcal{V}$ is the vertex set, $\mathcal{E}$ is the edge set, $\bt W$ is the weighted adjacency, and $\bt V$ is the self-loop matrix. The generalized graph Laplacian is $\bt L \defeq \bt D - \bt W + \bt V$, with $\bt D$ the degree matrix. Let $\bt L = \bt U \diag\bfl \bt U^\top$. Then $\bt U$ is the GFT basis.
\end{definition}
\begin{definition}[Rank-one updates] 
The set of symmetric rank-one updates of $\baseL$ is $\mathcal{P}(\baseL) \defeq \lbrace \bt L\in\mathbb{R}^{n\times n} \colon \bt L = \baseL+\rho \bt v\bt v^\top, \rho \in \mathbb{R}, \bt v \in \mathbb{R}^n \rbrace$. Elements are denoted by $\Lv\in\mathcal{P}(\baseL)$.
\label{def:rank-one-pert}
\end{definition}
\subsection{Rank-one relationships}
The next result states that the Laplacian decomposes into rank-one ``baby Laplacians'', one for each edge  contribution.
\begin{proposition}[\cite{batson2014twice}]
    Let $\bt e_j$ be the $j$th canonical vector. Then, the Laplacian $\bt L$ of an undirected graph can be written as
    \begin{equation}
        \bt L = \sum_{{(i, j)} \in  \mathcal{E}, \, i\not=j}\, w_{ij} (\bt e_i - \bt e_j)(\bt e_i - \bt e_j)^\top
        + \sum_{(i, i)\in\mathcal{E}}\,  w_{ii} \bt e_i\bt e_i^\top,
    \end{equation}
    where $w_{ij}$ denotes the edge weights, for $i, j = 1, \hdots, \vert \mathcal{V} \vert$. 
    \end{proposition}
There are two practical scenarios for the rank-one relationships of Def. \ref{def:rank-one-pert} (cf.~Fig.~\ref{fig:path_graph}).
\emph{(i) Adding a self-loop:} For a self-loop of strength $w$, we add a diagonal element to the Laplacian: $\bt L = \baseL+ w \,  \bt e_{i} \bt e_{i}^\top$. Transforms corresponding to path graphs with self-loops \cite{puschel_algebraic_2003} decorrelate optimally intra-predicted residuals under a Gauss-Markov model \cite{han2011jointly}.
\emph{(ii) Modifying an edge weight}: Modifying the weight of the $(i, j)$th edge by $w$ boils down to $\bt L = \baseL + w \, (\bt e_i - \bt e_j)(\bt e_i - \bt e_j)^\top$. Graphs with edge updates appear linked to Lapped transforms \cite{lu2019lapped} and graph learning \cite{lu2024adaptive, egilmez2020graph}.
\begin{definition}[DCT+] We denote by DCT+ the family of GFTs derived from rank-one updates of the Laplacian of the path graph.
\end{definition}
\begin{algorithm}[t]
 \caption{Fast DCT+, (cf.~\eqref{eq:fwd_map}), $\rho > 0$.}
 \label{algo:method_fwd}
 \begin{algorithmic}[1]
 \renewcommand{\algorithmicrequire}{\textbf{Input:}} 
 \REQUIRE Samples $\bt s_i$
  \STATE Set $\bt s_d = \mathrm{DCT}(\bt s)$ via DCT-II.
  \STATE Obtain $\bt s = \diag z \bt s_d$.
  \STATE Compute $p_n = \langle 1 / (\mu_n - \bfl), \bt s\rangle$.
  \STATE Set $h_j = (-1)^{j+1} s_j \sin(j\pi/n) / (1 - \tilde{\lambda}_j^2)$, $j = 1, \hdots, n-1$.
  \STATE Compute $\bt c = \mathrm{DST}(\bt h)$ via DST-I.
  \STATE Compute $\bt q = -1/2\, \mathrm{NFST}_{\tbfm}(\bt c) / \sin(n\acos(\tbfm))$.
  \STATE Compensate: $\bt p_{1:n-1} = (\bt q + s_1/\bfm_{1:n-1})$.
 \RETURN $\bt p_a = \diag a \bt p$.
 \end{algorithmic}
 \end{algorithm}
\section{Main results}
In the following, we let $\baseL = \bt U \diag\bfl\bt U^\top$ be a symmetric matrix, and $\Lv = \bt X\diag\bfm\bt X^\top \in \mathcal{P}(\baseL)$.
We show that the eigenvector bases before and after a rank-one update are related by a Cauchy matrix. 
To the best of our knowledge, this is the first result stating this relationship for arbitrary symmetric matrices, generalizing  recent ideas for diagonal matrices \cite{fasino_orthogonal_2023}.
\begin{theorem}
\label{thm:main_eq}
    The eigenvector basis $\bt X$ decomposes as
    \begin{equation}
        \label{eq:mapping}
        \bt X = \bt U \diag z \Clm\diag a,
    \end{equation}
    where $\bt z \defeq \bt U^\top \bt v$ and $\bt a\in\mathbb{R}^n$  normalizes each column. Moreover,
    \begin{equation}
        \label{eq:fwd_map}
        \bt X^\top = - \diag a \Cml \diag z \bt U^\top.
    \end{equation}    
\end{theorem}
When $\Lv$ is the Laplacian of a graph, this result offers a \emph{progressive factorization} of the GFT:  we apply first the unperturbed Laplacian's GFT and then the Cauchy matrix. 
 
The next theorem builds upon the previous result: we can rely on the existence of $O(n \log^2 n)$ \cite{gerasoulis1988fast} algorithms to compute matrix-vector products with Cauchy matrices to speed up the transform.

\begin{theorem} 
\label{th:fast_general}
    Assume that $\bt U$ has a $O(f(n))$ matrix-vector product algorithm. Then $\bt X$ has a 
    $O(n \log^2 n + f(n))$ matrix-vector product algorithm. 
\end{theorem}
In the context of GFTs, this theorem asserts that any transform derived from a rank-one update of a graph can be computed in $O(n\log^2 n)$ once the coefficients in the unperturbed transform domain are available. We now focus on the DCT+ case.
\begin{theorem} 
\label{th:fast_line}
    Let $\baseL$ be the Laplacian of the path graph. Then, an algorithm to compute the matrix-vector product $\bt X^\top \bt s_i$ with precision $\epsilon$ can be implemented in $O(n\log n + n\log(1/\epsilon))$ for any  $\bt s_i\in\mathbb{R}^n$.   
\end{theorem}

The proof of this Theorem is constructive: our forward transform is summarized in Algorithm~\ref{algo:method_fwd}. The backward transform can be implemented by step-wise inversion. Hence, there exist Fast DCT+ algorithms, up to some predetermined precision \cite{dutt1993fast}. 

Progressivity can also be combined with pruning when $k$ DCT+ have to be computed: we can keep only the low-frequency terms of the DCT and approximate the DCT+ via Cauchy matrices, as shown in Fig.~\ref{fig:diag_pruning}. The next result describes the complexity and error.
\begin{theorem}
\label{th:pruning_th} 
Let $\tilde{\bt C}(\bfm_r, \bfl) \doteq \diag {a_r} \bt C(\bfm_r, \bfl)\diag {z_r}$, and define $\tilde{\bt C}_{hl}(\bfm_r, \bfl)$ as its lower-left block of size $(n-c_p)\times c_p$, for $r = 1, \hdots, k$. Let $\hat{\bt C}(\bfm_r, \bfl)\in\mathbb{R}^{n\times n}$ be equal to $\tilde{\bt C}(\bfm_r, \bfl)$ in the $c_p\times c_p$ upper-left block and zero elsewhere.  Then,
$$
 \norm{\tilde{\bt C}(\bfm_r, \bfl)\bt s - \hat{\bt C}(\bfm_r, \bfl)\bt s}_2 \leq \norm{\vphantom{\tilde{\bt C}_{hl}(\bfm, \bfl)}\bt s_{c_p:n}}_2 + \norm{\tilde{\bt C}_{hl}(\bfm_r, \bfl)\bt s_{1:c_p}}_2,$$ with $\bt s \doteq \mathrm{DCT}(\bt s_i)$ and $\bt s_i\in\mathbb{R}^n$.  We can approximate the $k$ DCT+ with complexity $O(n\log n + (k-1)c_p^2)$.
\end{theorem}
The first term of the bound depends on the energy compaction in the DCT domain, while the second term is the leakage from low-frequencies in DCT domain to high-frequencies in the DCT+ domain. The  result follows from algebraic manipulations.

\begin{remark}
    The results and complexity claims above assume that the following variables are available from the setup stage: $\bfl$, $\bt U$ (or the DCT coefficients), $\bfm$ (obtained in $O(n)$ by solving a secular equation \cite{golub1973some}), $\bt z$ (obtained in $O(n\log n)$), and $\bt a$ (obtained in $O(n^2)$). 
\end{remark}
\section{Proofs and extensions}
\label{sec:prelem}
Let $\baseL$ and $\Lv$ be defined as before. Moreover, we assume  $\baseL$ does not have repeated eigenvalues and $\bt z = \bt U^\top\bt v$ does not have zero components. Otherwise, we  can apply deflation first \cite{gerasoulis1988fast}.

\subsection{Proofs of Theorem \ref{thm:main_eq} and Theorem \ref{th:fast_general}}
We state first an interleaving lemma.
\begin{lemma}[\cite{thompson1976behavior}]
\label{lemma:inter}
The eigenvalues of $\Lv$ satisfy
$\lambda_1 < \mu_1 < \lambda_2 < \cdots < \lambda_n < \mu_n$,
if $\rho>0$, reversing the order if $\rho < 0$.
\end{lemma}
Thus, $\mu_i \bt I - \diag{\bfl}$, for $i = 1, \hdots, n$, will be non-singular. Now, we relate the eigenvectors of $\baseL$ and  $\Lv$.
\begin{lemma}[\cite{bunch_rank-one_1978}]
\label{lemma:vec_decomp}
    Let $\bt z \defeq \bt U^\top \bt v$. Then,
        $\bt x_i = - a_i \bt U(\mu_i \bt I - \diag\bfl)^{-1}\bt z.$
\end{lemma}
To prove Theorem \ref{thm:main_eq}, let $\bt Y \defeq \bt U^\top \bt X$. Then, $Y_{j, i} = (a_iz_j)/(\lambda_j - \mu_i)$. Thus, $\bt Y = \diag z \bt C(\bfl, \bfm) \diag a$. Since $\bt X = \bt U \bt Y$ and $\Cml = -\Clm^\top$, results \eqref{eq:mapping} and \eqref{eq:fwd_map} follow.

Now, Theorem \ref{th:fast_general}  follows from \eqref{eq:mapping}: products with $\diag a$ and $\diag z$ can be computed in $O(n)$, products with $\bt U$ in $O(f(n))$ by assumption, and from Gerasoulis \cite{gerasoulis1988fast} we can compute matrix-vector products with Cauchy matrices in $O(n\log^2 n)$. 
\begin{figure}
    \centering
    \includegraphics[width=\linewidth]{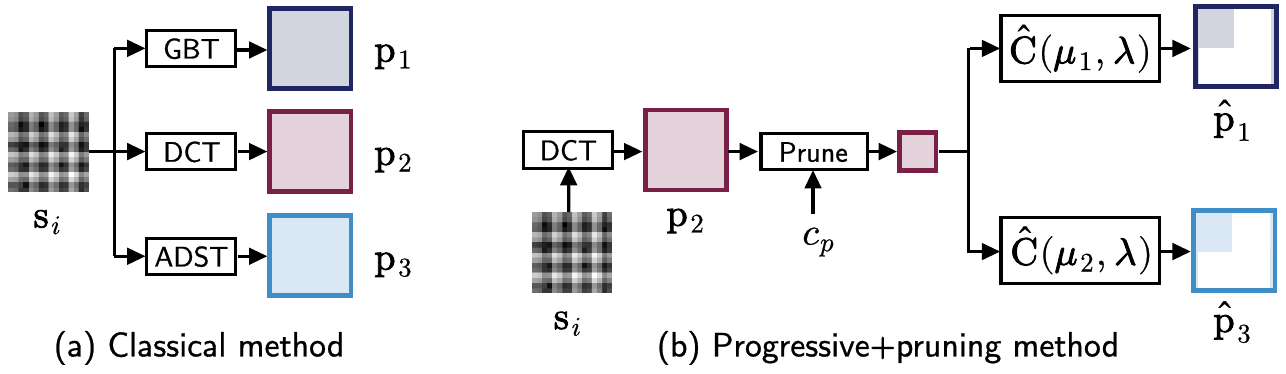}
    \caption{(a) Classical video coding approach to compute transforms for RDO. GBT is a graph-based tranform. In the best case, this approach has complexity $O(kn\log n)$. (b) Pruning method exploiting progressivity, with complexity $O(n\log n + (k-1)c_p^2)$.}
    \label{fig:diag_pruning}
    \vspace{-1em}
\end{figure}
\subsection{Proof of Theorem \ref{th:fast_line}}
\label{sec:line_proof}
We focus on the forward transform. The product with $\bt U^\top$ is in $O(n\log n)$ via DCT. Thus, we consider $\bt p = \Cml \bt s$, where $\bt s$ are the transform coefficients in DCT domain. 

Our algorithm has three stages: 1) pre-process the Cauchy matrix, 2) reduce the NMVP to a trigonometric polynomial, and 3) use the NUFFT \cite{dutt1993fast} to evaluate the polynomial. We assume $\rho>0$ in Def.~\ref{def:rank-one-pert}; the case $\rho < 0$ follows similarly.

\subsubsection{Pre-processing}
The pre-processing stage isolates $\lambda_1 = 0$ and $\mu_n$; the former cannot be transformed into a Chebyshev polynomial root, and the latter lies outside the evaluation range of Chebyshev polynomials (cf.~Lemma~\ref{lemma:inter}). The matrix-vector product can be written as:
\begin{equation}
    \bt p = \begin{bmatrix}
        1/\bfm_{1:n-1} & \bt C(\bfm_{1:n-1}, \bfl_{2:n}) \\ 
        1 / \mu_n & 1 / (\mu_n - \bfl_{2:n}^\top)
    \end{bmatrix}\begin{bmatrix}
        s_1 \\ 
        \bt s_{2:n}
    \end{bmatrix}.
\end{equation}
Note that $p_n = \langle 1 / (\mu_n - \bfl), \bt s\rangle$ and $\bt p_{1:n-1} = \bt C(\bfm_{1:n-1}, \bfl_{2:n})\bt s_{2:n} + s_1/\bfm_{1:n-1}$. Let $\bt q = \bt C(\bfm_{1:n-1}, \bfl_{2:n})\bt s_{2:n}$.

The roots of $U_{n-1}(\cdot)$, the Chebyshev polynomial of the second kind and order $n-1$, satisfy $y_k = 1 - \lambda_{k+1}/2$, for $k = 1, \hdots, n-1$. Let us introduce $\tbfm \defeq 1 - \bfm_{1:n-1}/2$ and $\tbfl \defeq 1 - \bfl_{2:n}/2$. By properties of Cauchy matrices, $\bt q = -1/2 \, \bt C(\tbfm, \tbfl)\bt s_{2:n}$. 

\subsubsection{Reduction to trigonometric polynomials}
We rewrite the matrix-vector product as a trigonometric polynomial. Let $q_i = f(\tilde{\mu_i}) = \sum_{j=1}^{n-1} \, s_j / (\tilde{\mu_i} - \tilde{\lambda_j})$ for $i = 1, \hdots, n-1$. Gastinel \cite{gastinel1960inversion} stated the following result about $f(x)$.
\begin{proposition}[\cite{gastinel1960inversion}]
    Let $g(x) \defeq \prod_{j = 1}^{n-1}(x - \tilde{\lambda}_j)$, and define $h(x)$ as a polynomial of degree $n-2$ such that $h(\tilde{\lambda}_j) = s_j g'(\tilde{\lambda}_j)$, for $j = 1, \hdots, n-1$. Then, $f(x) = h(x)/g(x)$.
\end{proposition}
By leveraging Chebyshev polynomials of the second kind, it can be shown that $g(x) = d_nU_{n-1}(x)$ and $g'(\tilde{\lambda}_j) = d_n(-1)^{j+1}n/(1-\tilde{\lambda}_j^2)$, for $j = 1, \hdots, n-1$, where $d_n = 2^{-(n-1)}$. Since $h$ has degree $n-2$, we can interpolate it from the $n-1$ samples $\lbrace h(\tilde{\lambda}_j) \rbrace_{j = 1}^{n-1}$:
\begin{lemma}
\label{lemma:dst}
    Let $h(x) = \sum_{k = 1}^{n-1}\,c_\ell U_{\ell-1}(x)$. Then, 
    \begin{equation*}
    c_\ell = 2/n\sum_{j = 1}^{n-1} \, h(\tilde{\lambda}_j)\sin(\pi j/n) \, \sin(\ell j\pi /n), \ \ell = 1, \hdots, n-1,
    \end{equation*}
    which can be computed in $O(n\log n)$ via DST-I with FFT.
\end{lemma}
\begin{proof}
Chebyshev polynomials satisfy \cite{mason2002chebyshev}, for $l + m \leq 2n-3$, $
    \sum_{k = 1}^{n-1}\,(1-\tilde{\lambda}_k^2)
    U_l(\tilde{\lambda}_k) U_m(\tilde{\lambda}_k) = n/2 \, \delta[l - m]$,
where $\delta$ is the Kronecker delta. Thus, $c_\ell = 2/n \sum_{j = 1}^{n-1}\, (1-\tilde{\lambda}_j^2)h(\tilde{\lambda}_j) U_{\ell-1}(\tilde{\lambda}_j)$. The result follows from $1-\tilde{\lambda}_j^2 = \sin^2(j\pi / n)$.
\end{proof}
\begin{figure*}
		\centering
        \includegraphics[width=\linewidth]{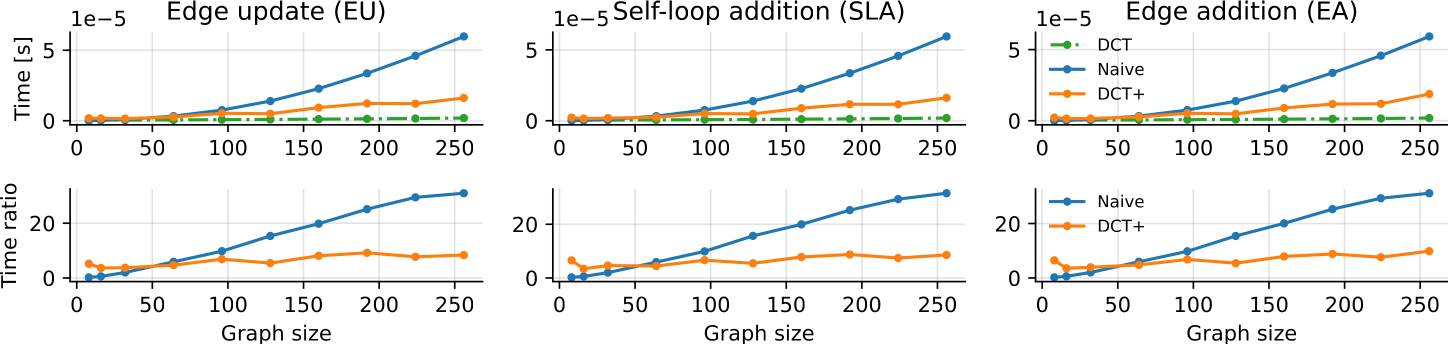}
      \caption{Top row: average runtime for transforming $10000$ AR($0.99$) signals of size $\lbrace 8, 16, 32, 64, 96, 128, 160, 192, 224, 256\rbrace$ using the Fast DCT+ (DCT+), NMVP for DCT+ (Naive), and as a baseline, the time to compute the DCT. Bottom row: results normalized by the DCT time. We consider the three graphs in Fig.~\ref{fig:path_graph}(b-d), letting $\alpha = 1.5$. For edge updates, we update the edge $(2, 3)$. For edge addition, we add $(3, 5)$.}
      \label{fig:fast_results}
      \vspace{-0.5em}
\end{figure*}
Thus, we can obtain an expression for $q_i$ as a scaled trigonometric polynomial in $\acos(\tilde{\mu}_i)$:
\begin{equation}
\label{eq:trigo}
    q_i = \frac{\sum_{k = 1}^{n-1}\, c_k\sin \l k\acos(\tilde{\mu_i})\r}{d_n\sin(n\acos(\tilde{\mu_i}))}, \quad i = 1, \hdots, n-1.
\end{equation}
\subsubsection{Non-uniform Fast Fourier transform}
To evaluate \eqref{eq:trigo}, we use the sinusoidal version of the non-uniform fast Fourier transform (NUFFT) \cite{dutt1993fast} or NFST \cite{plonka2018numerical}. NFST computes \eqref{eq:trigo} at non-uniform sampling points with precision $\epsilon$ and complexity $O(n\log n+n\log(1/\epsilon))$. Hence, Theorem~\eqref{th:fast_line} follows.
\subsection{Extension to other graphs and rank-$k$ updates}
The result is not limited to DCT+ For instance, for the ADST graph, $g(x) = T_n(x)$ (the Chebyshev polynomial of the first kind) and similar techniques can be applied. If $\bar{\mu}_i$ has a trigonometric structure (e.g., DTTs) \cite{lu2022dct}, equation \eqref{eq:trigo} can be evaluated in $O(n\log n)$. This creates relationships between DTTs \cite{puschel_algebraic_2003}.

Rank-$k$ updates preserve Cauchy structures; the proof follows by induction. Thus, if the base transform is in $O(f(n))$, the transform for any rank-$k$ update is in $O(kn\log^2n + f(n))$.
\begin{corollary}
    Let $k\in \mathbb{Z}^+$, $\baseL = \bt U \diag\bfl\bt U^\top$, and $\bt L_k \defeq \baseL + \sum_{j = 1}^k\, \rho_j \bt v_j \bt v_j^\top = \bt X_k\diag{\bfm_k}\bt X_k^\top$. Define $\bt z_k = \bt U^\top \bt v_k$.Then,
        $\bt X_k = \bt X_{k-1} \diag{z_k}\bt C(\bfm_{k-1}, \bfm_{k})\diag{a_k}$,
    with $\bt a_k$ normalizing the columns of $\bt X_k$. 
\end{corollary}
\vspace{-1em}
\begin{table}[t]
    \centering
    \begin{tabular}{lcccccc}
        \toprule \textbf{Size}
          & \textbf{8} & \textbf{16} & \textbf{32} & \textbf{64} & \textbf{128} & \textbf{256} \\ 
        \midrule
            \textbf{EU} & 142.5 & 135.9 & 156.3 & 136.4 & 110.0 & 114.8   \\	
            \textbf{SLA} & 133.7 & 120.5 & 142.2 & 129.5 & 138.3 & 109.3  \\	
            \textbf{EA} & 117.5 & 126.8 & 137.2 & 104.4 & 124.7 & 102.7    \\	
        \bottomrule
    \end{tabular}    
    \caption{Average SNR (in dB) at the output of the Fast DCT+ in Fig.~\ref{fig:fast_results}. The error is measured against the result of NMVP. EU: edge update, SLA: self-loop addition, EA: edge addition.}\label{tab:table_times}
    \vspace{-1em}
\end{table}	
\begin{figure}
    \centering
\includegraphics[width=\linewidth]{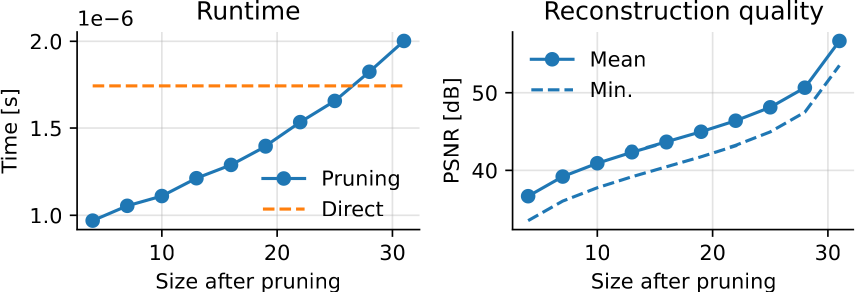}   
    \caption{Aggregated runtime and reconstruction quality versus the number of components after pruning. We consider DCT, ADST, and the transform of a path graph with a self-loop of $1.5$. The direct method computes each transform with its fastest implementation. Signals are rows from blocks of size $32\times 32$, obtained from  $20$ images of the CLIC dataset \cite{CLIC2022}. Reconstruction quality is measured against the output of the direct method; we compute the mean-squared error and take mean and maximum among all signals.}
    \label{fig:results}
    \vspace{-1.5em}
\end{figure}
\section{Numerical Experiments}
\label{sec:exper}
We run simulations in a Intel Xeon CPU E5-2667 with $16$ cores. The DCT/DST are implemented using FFTW \cite{frigo2005fftw}, while the NFST is taken from the NFFT3 package \cite{keiner2009using}. 

\medskip

\noindent \textbf{Fast transforms.} We compare Fast DCT+ with NMVP for DCT+ for different graph sizes. We show our runtime results (Fig.~\ref{fig:fast_results}) and the approximation errors (Table~\ref{tab:table_times}). Fast DCT+ is numerically accurate ($100$ dB); we also observe that the runtime of the Fast DCT+ is approximately  $8$ times the runtime of the DCT, for all graph sizes, verifying the dependency with $O(n\log n)$. For graph sizes larger than $64$, we obtain computational savings with respect to NMVP. 

\medskip

\noindent \textbf{Pruning.} We consider a video coding scenario \cite{han2021technical} with three transforms (cf.~Fig.~\ref{fig:diag_pruning}). We compare the direct approach to a simplified pruning approach, where, for each signal, we compute the corresponding quadratic Laplacian form from the path graph. 
We compare this value with a threshold to decide whether the pruning strategy will be applied or not. The runtime and the error is averaged over all the realizations. The trade-off between accuracy and speed is controlled by the size of the signal after pruning; the threshold is fixed. We do not consider parallel computations for either the transforms in the direct approach or those in the second branch of the pruning scheme. All the products with Cauchy matrices are implemented using NMVP. Results show (Fig.~\ref{fig:results}) speedups in exchange for some error, which is smaller than the typical quantization noise.  

\section{Conclusion}
\label{sec:majhead}
We introduced a theoretical framework to derive fast GFTs for rank-one updates of the path graph, which we termed Fast DCT+. We first provided a structured and progressive factorization for the DCT+ transforms in terms of Cauchy matrices. By leveraging the structure of these matrices, DCT+ can be implemented with precision $\epsilon$ and complexity $O(n\log n+n\log(1/\epsilon))$, offering significant computational savings for larger graph sizes with SNRs above $100$ dB. By exploiting progressivity, we can reduce computations with pruning strategies, achieving runtime improvements while preserving PSNRs above $40$ dB. Future work will focus on further experimental validation, other graph types \cite{fernandez2023image}, and integration with graph learning \cite{pavez2017learning}.

\clearpage

\bibliographystyle{IEEEbib}
\bibliography{ICASSP_2025}

\begin{thebibliography}{10}

\bibitem{ortega2018graph}
A. Ortega, P. Frossard, J. Kova{\v{c}}evi{\'c}, J.~M. Moura, and P. Vandergheynst,
\newblock ``Graph signal processing: Overview, challenges, and applications,''
\newblock {\em Proc. IEEE}, vol. 106, no. 5, pp. 808--828, 2018.

\bibitem{brigham1988fast}
E.~O. Brigham,
\newblock {\em The fast {Fourier} transform and its applications},
\newblock Prentice-Hall, Inc., 1988.

\bibitem{strang1999discrete}
G. Strang,
\newblock ``The discrete cosine transform,''
\newblock {\em SIAM Rev.}, vol. 41, no. 1, pp. 135--147, 1999.

\bibitem{puschel_algebraic_2003}
M. Püschel and J.~M.~F. Moura,
\newblock ``The algebraic approach to the discrete cosine and sine transforms and their fast algorithms,''
\newblock {\em SIAM J. Comput.}, vol. 32, no. 5, pp. 1280--1316, Jan. 2003.

\bibitem{ekambaram2013circulant}
V.~N. Ekambaram, G.~C. Fanti, B. Ayazifar, and K. Ramchandran,
\newblock ``Circulant structures and graph signal processing,''
\newblock in {\em Proc. IEEE Int. Conf. Image Process.} IEEE, 2013, pp. 834--838.

\bibitem{le2017approximate}
L. Le~Magoarou, R. Gribonval, and N. Tremblay,
\newblock ``Approximate fast graph {Fourier} transforms via multilayer sparse approximations,''
\newblock {\em IEEE Trans. Signal Process. over Net.}, vol. 4, no. 2, pp. 407--420, 2017.

\bibitem{frerix2019approximating}
T. Frerix and J. Bruna,
\newblock ``Approximating orthogonal matrices with effective givens factorization,''
\newblock in {\em Int. Conf. on Mach. Learn.} PMLR, 2019, pp. 1993--2001.

\bibitem{lu2019fast}
K.-S. Lu and A. Ortega,
\newblock ``Fast graph {Fourier} transforms based on graph symmetry and bipartition,''
\newblock {\em IEEE Trans. Signal Process.}, vol. 67, no. 18, pp. 4855--4869, 2019.

\bibitem{han2013butterfly}
J. Han, Y. Xu, and D. Mukherjee,
\newblock ``A butterfly structured design of the hybrid transform coding scheme,''
\newblock in {\em Proc. Pict. Cod. Symp.} IEEE, 2013, pp. 17--20.

\bibitem{malvar1990lapped}
H.~S. Malvar,
\newblock ``Lapped transforms for efficient transform/subband coding,''
\newblock {\em IEEE Trans. on Acoust., Speech, and Signal Process.}, vol. 38, no. 6, pp. 969--978, 1990.

\bibitem{bross2021overview}
B. Bross, Y.-K. Wang, Y. Ye, S. Liu, et~al.,
\newblock ``Overview of the versatile video coding {(VVC)} standard and its applications,''
\newblock {\em IEEE Trans. on Circ. and Sys. for Video Tech.}, vol. 31, no. 10, pp. 3736--3764, 2021.

\bibitem{pavez2017learning}
E. Pavez, A. Ortega, and D. Mukherjee,
\newblock ``Learning separable transforms by inverse covariance estimation,''
\newblock in {\em Proc. IEEE Int. Conf. Image Process.} IEEE, 2017, pp. 285--289.

\bibitem{sole_segui_communication}
A. Segui and J. Sole,
\newblock {\em Private communication},
\newblock 2024.

\bibitem{han2021technical}
J. Han, B. Li, D. Mukherjee, C.-H. Chiang, et~al.,
\newblock ``A technical overview of {AV1},''
\newblock {\em Proc. IEEE}, vol. 109, no. 9, pp. 1435--1462, 2021.

\bibitem{lu2019lapped}
K.-S. Lu and A. Ortega,
\newblock ``Lapped transforms: A graph-based extension,''
\newblock in {\em Proc. IEEE Int. Conf. Acoust., Speech, and Signal Process.} IEEE, 2019, pp. 5401--5405.

\bibitem{lu2024adaptive}
W.-Y. Lu, E. Pavez, A. Ortega, X. Zhao, and S. Liu,
\newblock ``Adaptive online learning of separable path graph transforms for intra-prediction,''
\newblock in {\em Proc. Pict. Cod. Symp.} IEEE, 2024, pp. 1--5.

\bibitem{egilmez2020graph}
H.~E. Egilmez, Y.-H. Chao, and A. Ortega,
\newblock ``Graph-based transforms for video coding,''
\newblock {\em IEEE Trans. Image Process.}, vol. 29, pp. 9330--9344, 2020.

\bibitem{bunch_rank-one_1978}
J.~R. Bunch, C.~P. Nielsen, and D.~C. Sorensen,
\newblock ``Rank-one modification of the symmetric eigenproblem,''
\newblock {\em Numer. Math.}, vol. 31, no. 1, pp. 31--48, Mar. 1978.

\bibitem{golub1985trummer}
G.~H. Golub,
\newblock ``Trummer problem,''
\newblock {\em SIGACT News}, vol. 17, no. 2, pp. 17.2, 1985.

\bibitem{gerasoulis1988fast}
A. Gerasoulis,
\newblock ``A fast algorithm for the multiplication of generalized {Hilbert} matrices with vectors,''
\newblock {\em Math. of Comp.}, vol. 50, no. 181, pp. 179--188, 1988.

\bibitem{fasino_orthogonal_2023}
D. Fasino,
\newblock ``Orthogonal {Cauchy}-like matrices,''
\newblock {\em Numer. Algor.}, vol. 92, no. 1, pp. 619--637, Jan. 2023.

\bibitem{dutt1993fast}
A. Dutt and V. Rokhlin,
\newblock ``Fast {Fourier} transforms for nonequispaced data,''
\newblock {\em SIAM Journal on Scient. Comp.}, vol. 14, no. 6, pp. 1368--1393, 1993.

\bibitem{ortega1998rate}
A. Ortega and K. Ramchandran,
\newblock ``Rate-distortion methods for image and video compression,''
\newblock {\em IEEE Signal Process. Mag.}, vol. 15, no. 6, pp. 23--50, 1998.

\bibitem{han2011jointly}
J. Han, A. Saxena, V. Melkote, and K. Rose,
\newblock ``Jointly optimized spatial prediction and block transform for video and image coding,''
\newblock {\em IEEE Trans. Image Process.}, vol. 21, no. 4, pp. 1874--1884, 2011.

\bibitem{lengwehasatit2004scalable}
K. Lengwehasatit and A. Ortega,
\newblock ``Scalable variable complexity approximate forward {DCT},''
\newblock {\em IEEE Trans. on Circ. and Syst. for Video Tech.}, vol. 14, no. 11, pp. 1236--1248, 2004.

\bibitem{gastinel1960inversion}
N. Gastinel,
\newblock ``Inversion d’une matrice generalisant la matrice de {Hilbert},''
\newblock {\em Chiffres}, vol. 3, pp. 149--152, 1960.

\bibitem{batson2014twice}
J. Batson, D.~A. Spielman, and N. Srivastava,
\newblock ``Twice-{Ramanujan} sparsifiers,''
\newblock {\em SIAM Rev.}, vol. 56, no. 2, pp. 315--334, 2014.

\bibitem{golub1973some}
G.~H. Golub,
\newblock ``Some modified matrix eigenvalue problems,''
\newblock {\em SIAM Rev.}, vol. 15, no. 2, pp. 318--334, 1973.

\bibitem{thompson1976behavior}
R.~C. Thompson,
\newblock ``The behavior of eigenvalues and singular values under perturbations of restricted rank,''
\newblock {\em Linear Alg. and its Apps.}, vol. 13, no. 1-2, pp. 69--78, 1976.

\bibitem{mason2002chebyshev}
J.~C. Mason and D.~C. Handscomb,
\newblock {\em Chebyshev polynomials},
\newblock Chapman and Hall/CRC, 2002.

\bibitem{plonka2018numerical}
G. Plonka, D. Potts, G. Steidl, and M. Tasche,
\newblock {\em Numerical {Fourier} analysis},
\newblock Springer, 2018.

\bibitem{lu2022dct}
K.-S. Lu, A. Ortega, D. Mukherjee, and Y. Chen,
\newblock ``{DCT} and {DST} filtering with sparse graph operators,''
\newblock {\em IEEE Trans. Signal Process.}, vol. 70, pp. 1641--1656, 2022.

\bibitem{CLIC2022}
``5th {Challenge on Learned Image Compression dataset},'' Online, June 2022.

\bibitem{frigo2005fftw}
M. Frigo and S.~G. Johnson,
\newblock ``The design and implementation of {FFTW3},''
\newblock {\em Proc. IEEE}, vol. 93, no. 2, pp. 216--231, 2005.

\bibitem{keiner2009using}
J. Keiner, S. Kunis, and D. Potts,
\newblock ``Using {NFFT} 3---a software library for various nonequispaced fast {Fourier} transforms,''
\newblock {\em ACM Trans. on Math. Soft.}, vol. 36, no. 4, pp. 1--30, 2009.

\bibitem{fernandez2023image}
S. Fern{\'a}ndez-Mendui{\~n}a, E. Pavez, and A. Ortega,
\newblock ``Image coding via perceptually inspired graph learning,''
\newblock in {\em Proc. IEEE Int. Conf. Image Process.} IEEE, 2023, pp. 2495--2499.

\end{thebibliography}
\end{document}